\documentclass{amsart}
\usepackage{amsmath}
\usepackage{amsthm}

\newtheorem{theorem}{Theorem}[section]

\usepackage{listings}
\usepackage{courier}
\usepackage{color}

\begin{document}
\title{A General Class Of Collatz Sequence and Ruin Problem} 

\author{Nabarun Mondal}
\address{D.E.Shaw \& Co., Hyderabad }
\email{mondal@deshaw.com}
\thanks{
Special thanks to our colleagues at Microsoft India R\&D : Dhrubajyoti Ghosh \& Ramkishore Bhattacharyya \& Samira Japanwala
Microsoft India R \& D, Hyderabad \& D.E.Shaw India, Hyderabad }

\author{Partha P. Ghosh}
\address{Microsoft India R\&D, Hyderabad }
\email{parthag@microsoft.com}
\thanks{
}


\date{November 8 2010}

\dedicatory{This paper is dedicated to our parents and family.}


\begin{abstract}

In this paper we show the probabilistic convergence of the original Collatz (3n + 1) (or Hotpo) sequence to unity. A generalized form of the Collatz sequence (GCS) is proposed subsequently. Unlike Hotpo, an instance of a GCS can converge to integers other than unity. A GCS can be generated using the concept of an abstract machine performing arithmetic operations on different numerical bases. Original Collatz sequence is then proved to be a special case of GCS on base 2. The stopping time of GCS sequences is shown to possess remarkable statistical behavior. We conjecture that the Collatz convergence elicits existence of attractor points in digital chaos generated by arithmetic operations on numbers. We also model Collatz convergence as a classical ruin problem on the digits of a number in a base in which the abstract machine is computing and establish its statistical behavior. Finally an average bound on the stopping time of the sequence is established that grows linearly with the number of digits.

\end{abstract}

\maketitle

\section {Description Of The Problem : Collatz Conjecture}

The original Collatz function is defined as :- 

\begin{equation}\label{orig_Collatz}
f(n) = 
\left\{
\begin{array}{lr}
3n+1 \; ,              &  n \equiv 1 \; mod \; 2  \\
\frac{n}{2} \; ,       & n \equiv 0 \; mod \; 2 
\end{array}
\right. 
\end{equation}

Beginning with any positive integer, and taking the result of \eqref{orig_Collatz} 
at each step as the input at the next, 
a sequence can be formed, as explained below:-

Assume $x = f(n)$. Applying $f$ on $x$ itself to get  $f(f(n))$.
For shorter notation, it can be written as $f^2(n)$.
In general after applying this function ``k'' number of times, can be written as $f^k(n)$.

According to the Collatz conjecture repeated application of this process will eventually converge to `1', regardless of which positive integer is chosen initially.	
In other words the Collatz conjecture states there exist an integer $\sigma(n)$, however large, for every positive integer `n' , such that $ f^{\sigma(n)}(n) = 1 $.This $\sigma(n)$ is called the Stopping Time.

\section{Contemporary Works : Inspiration for our work}
Lagarias \cite{lagarias} describes the history of the Collatz problem and surveys all the major literature.
Lagarias \cite{lagarias} is taken as the primary reference for this paper. 
Below mentioned observations are inspiration for this paper, as presented in the subsections.

\subsection {Collatz Sequence In Base 2}

It is a well known fact that computing Collatz sequence
can be seen as an abstract computer computing the sequence in base `2', whose only operations would be left shift, right shift and addition, as explained below.

\newcounter{qcounter}
The machine will perform the following  steps on any odd number `n' until the number reduces to  ``1'' :-

\begin{list}{\arabic{qcounter}:~}{\usecounter{qcounter}}
\item Left shift `n' by one bit : thus giving 2n;
\item Add (1) to the original number, `n'  by binary addition \\ (giving 2n+n = 3n);
\item Add 1 to the right of the new number  in (2) by binary addition.
\item Remove all trailing ``0''s (i.e. repeatedly divide by two until the result is odd). This is the right shift operation.
\end{list}

It is also known from Lagarias \cite{lagarias} that the odd numbers generated in the sequence are generated in `almost' random order and that the convergence can be treated probabilistically.  

\subsection {General Collatz Class}

Several authors (M\"{o}ller \cite{moller}) have investigated the range of validity of the result that has a finite stopping time for almost all integers ``n'' by considering more general classes of periodicity linear functions. One such class  consists of all functions  which are given by 

\begin{equation}\label{gen_Collatz_2}
U(n) = 
\left\{
\begin{array}{lr}
\frac{n}{b} \; ,                     & n \equiv 0 \; mod \; b \\ 
\frac{mn - r}{b} \;, \; r = mn \; mod \; b  \; ,  & otherwise 
\end{array}
\right.
\end{equation}

It has been shown in M\"{o}ller \cite{moller} that iff 

$$
m < b^{\frac{b}{b-1}} 
$$
then sequence generated by \eqref{gen_Collatz_2} would converge.

\subsection{Stopping Time}

Assuming that the  Conjecture is true, one can consider the problem of determining the  expected stopping time function. Crandall \cite{crandall} and Shanks \cite{shanks} were guided by probabilistic heuristic arguments to conjecture that for large set of large random integers
the ratio of stopping time to natural logarithm of the integer should approach a constant limit:- 

\begin{equation}\label{gen_Collatz_stopping_time}
\frac{\sigma(n)}{ln(n)} = 2 \left( ln \left (\frac{4}{3} \right )\right)^{-1} 
\end{equation}

where $\frac{\sigma(n)}{ln(n)}$ is sampled over large set of large integers.

\section{The Work}

\subsection{Overview}
In present work an intuitive general formula for Collatz like sequences is described, basing the formula
on an abstract computing machine operating on different bases. 

Based on the digits on the tape of the abstract machine it has been conjectured
that the Collatz like problems are examples of presence of digit-wise chaos in arithmetic operations.

To prove the digit-wise chaos, a very simple pseudorandom bit generator is presented which is based upon 
Collatz sequence and which passes all of the diehard battery of randomness tests.

The convergence criterion for Collatz like sequence then become more probable right shift operation than 
left shift operations of the abstract machine, and which becomes statistically decreasing number of digits,
which reduces the problem into classical ruin problem on number of digits.

Finally, based upon simple probability, an approximation of the stopping time ratio is established like in 
\eqref{gen_Collatz_stopping_time}.

\subsection{The Implicit Base Assumption}

Let us rewrite, 

$$  3n+1 = 2n + (n+1)  $$

And generalize it as Left-shift, followed by addition of an even number ($n+1$ is an even number).
This deletes the last bit of the binary representation of the original number $n$. 

It evidently means that adding $n+1$ is the trick to round off the last digit of the modified number $n$ to $0$ in base $2$.
This insight is duly confirmed by the observation that,

$$  3n+3 = 2n + (n+2)  $$

also converges to \{1,3\} depending upon  whether or not $n$ is a power of $2$. 
If it were not, it would converge to $3$, else it would converge to $1$.
It has been known for sometime that these general $3n+k$ forms are nothing but transforms (scaling) of the 
original $3n+1$ sequence, where $k$ is an odd number.

Based upon these observations, it can be suggested that the Collatz Sequence has base $2$ 
implicit within it. The breaking up of $3n + 1 = 2n + (n+1)$ is a clue to this insight.

If this \emph {``theory of implicit base''}  holds, then 
\emph { on any base, there would exist a general formula of a sequence which would be based upon the left-shift and right-shift paradigm only}. 

\subsection {The Generalized Collatz Formula}

We found that if the general formula has to exist, it has to be of the form :-

\begin{equation}\label{gen_Collatz_1}
f(n) = 
\left\{
\begin{array}{lr}
\frac{n}{b} \; ,                     & n \equiv 0 \; mod \; b \\
(b+1)n+ (b - n \; mod \; b ) \; ,    & otherwise 
\end{array}
\right.
\end{equation}

the $(b+1)n+ (b - n \; mod \; b )$ part of course is $bn + ( n + (b - n \; mod \; b ) )$.
That is the basic algorithm that evolved from the concept of left and right shift operations in an abstract machine computing Generalized Collatz sequence in base $b$.

This sequence is in fact  similar to the sequence mentioned in \cite{moller} as in \eqref{gen_Collatz_2}.
We note that:- 
$$
b + 1 < b^{\frac{b}{b-1}}, \forall b \ge 2 
$$
so, the convergence criterion is satisfied.  

The findings on the $f(n)$ are summarized  in the  table [\ref{table:results}].

\begin{table}[ht]
\caption{Results of the Tests on Generalized Collatz Sequence}
\centering
\begin{tabular}{ | l | l |  l |  }
\hline\hline
Base (b) & Convergence Points & Tested Up-to  \\ [0.5ex]
%heading
\hline
2   & 1              & Maximum Integer     \\ \hline
3   & 1,7,?          & 100000              \\ \hline
5   & 1              & 1354093827          \\ \hline
7   & 1              & 386581748           \\ \hline
11  & 1,642,?        & 100000              \\ \hline
13  & 1              & 24000000000         \\ \hline
17  & 1,79,?         & 100000              \\ \hline
19  & 1              & 340000000           \\ \hline
23  & 1,82,?         & 100000              \\ \hline
29  & 1,111,?        & 100000              \\ \hline
31  & 1,389,?        & 100000              \\ \hline
37  & 1              & 100000000           \\ \hline
41  & 1              & 99999999            \\ \hline
43  & 1              & 99999999            \\ \hline
\hline
\end{tabular}
\label{table:results}
\end{table}

\subsection {A Simpler Collatz Sequence }
A simpler and computationally faster formulation of Collatz Type sequence would be:-
\begin{equation}\label{generalized_simplified_Collatz}
f(n) = 
\left\{
\begin{array}{lr}
\frac{n}{b} \; ,                     & n \equiv 0 \; mod \; b \\
\lfloor\frac{(b+1)n}{b} \rfloor \; ,    & otherwise 
\end{array}
\right.
\end{equation}

This can be computed in the fastest way in a digital computer, 
because no modulus operation is required.
Our original generic sequence formula (\ref{gen_Collatz_1}) was nothing but:-
$$
f(n) = 
\left\{
\begin{array}{lr}
\frac{n}{b} \; ,                               & n \equiv 0 \; mod \; b \\
\lfloor\frac{(b+1)n}{b} \rfloor +1 \; ,    & otherwise 
\end{array}
\right.
$$

The simplified actual Collatz Sequence then can be written as
\begin{equation}\label{original_simplified_Collatz}
f(n) = 
\left\{
\begin{array}{lr}
\frac{n}{2} \; ,                     & n \equiv 0 \; mod \; 2 \\
\lfloor\frac{3n}{2} \rfloor \; ,    & otherwise 
\end{array}
\right.
\end{equation}

The table [\ref{table:results_special}] shows some of the results of the experiments on small bases.

\begin{table}[!ht]
\caption{Some Results on Simplified Generalized Collatz Sequence}
\centering
\begin{tabular}{ | l | l |  l |  }
\hline\hline
Base & Convergence Points & Tested Up-to  \\ [0.5ex]
%heading
\hline
2   & 1,5,17,?             & 100000         \\ \hline
3   & 1,2,22,?             & 100000         \\ \hline
5   & 1,2,3,4,57,?         & 100000         \\ \hline
7   & 1,2,3,4,5,6,?        & 100000         \\ \hline
\hline
\end{tabular}
\label{table:results_special}
\end{table}

\subsection {Pseudorandom Sequence Generator and The Generalized Collatz Sequence}

A Linear Congruential Generator (LCG) represents one of the oldest and best-known algorithm for the generation of pseudorandom numbers.
The generator is defined by the recurrence relation:-

$$
X_{n+1} = ( aX_n + c ) mod \; m
$$
where:- \\ 
$m > 0$ is the modulus \\
$a, 0 < a < m$, is the multiplier \\
$c, 0 \le c < m$, is the increment \\
$X_0, 0 < X_0 < m$, is the start value \\
$n \ge 0 $ is the iteration number \\

Assume we are in base $2$, and we have a tape holding $k$ bits of  information. Treat the tape as a register having $k$ initialized bits.\\
Assume also $d_i = 0 \;, \forall i \ge k $.

Then we can define the recurrence relation on time, $t$ (present), and $t+1$, the next time interval as follows:-
$$
d_1(t+1) = d_1(t) 
$$
$$
d_2(t+1) =  (  d_1(t+1) + d_2(t) + 1  )mod \; 2 
$$
$$
d_3(t+1) =  (  d_2(t+1) + d_3(t)   )mod \; 2 
$$
$$
...
$$
$$
d_{k-1}(t+1) =  ( d_{k-2}(t+1) + d_{k-1}(t) ) mod \; 2 
$$
$$
d_k = Carry \; digit 
$$ 

Specifically 

$$
d_{k-1}(t+1) =  ( \sum_{i=1}^{k-1}{d_i}(t)  ) mod \; 2
$$  

We can see it is nothing but a LCG, applied on $d_i$,  $(i-1)$ times, to generate the $d_i$ bit. 
Specifically:-

$$
d_{k-1}(t+1) =  ( d_{k-1}(t)  +   \sum_{i=1}^{k-2}{d_i}(t)   ) mod \; 2
$$

This is an LCG where :- \\
$a=1$ and $c = (\sum_{i=1}^{k-2}{d_i}(t)  ) $, where `$c$' itself is generated using the same LCG, on lower bits.

But this recurrence relation is the prescription of the operation $3n + 1$ in binary digits on an abstract machine. 
Replacing $ mod \; 2$ with $ mod \; b$ would be 
the prescription of calculating  $(b+1)n+ (b - n \; mod \; b )$ in base $b$ on an abstract machine, with set of LCGs:-
\begin{equation}\label{Collatz_lcg}
d_{k-1}(t+1) =  ( d_{k-1}(t)  +   \sum_{i=1}^{k-2}{d_i}(t)   ) mod \; b
\end{equation}
This observation is in line with Feinstein \cite{feinstein} that the standard mathematical way of proving this conjecture may not possible.
Also note that this observation moves the Collatz Problem - from number theory to Chaos Theory, and shows all generations
of numbers generated by Collatz function would behave randomly. Similarities with cellular automata is apparent, 
and is discussed with a cellular automata in base 6 in Wolfram \cite{anks} which simulates the original Collatz Function. 

It is also stated in Wolfram \cite{anks} that arithmetic operations can have digit-wise chaos,
and we hypothesize that the convergence on generalized Collatz type sequences might well be statistical, and the system behaves rather chaotically in higher number of digits.

To test the hypothesis we did pass numbers with highly ordered binary representation, for example :-

\begin{itemize}
\item  { $n = 2^k - 1$, this immediately reduces to $3^k -1$, and  elementary proof exists for this reduction. 
The system $f$ from section 1 can be written as $\frac{3}{2}(n+1) -1$ when n is odd number, and hence $f^k = {(\frac{3}{2})}^k (n+1) - 1 $ if all the $k$ iteration produced odd number. 
The reduction to $3^k -1$ immediately follows from there.}
\item  { $n = 2^k + 1$ where it reduces to $2^k + 2^{k-1} + 2 + 1$. In binary form, it looks like $110000...00011$. If the ``11'' in the left and right are separated enough, 
so that they are isolated in the next $3n+1$ operation, (i.e. carry does not reach the leftmost ``11'') then it would reduce to $100100...001$ 
reducing the $n_0$ by 2, and increasing the $n_1$ by 1.  For example: $ 1000001 \rightarrow 1100011 \rightarrow 100101000 \rightarrow 100101 $}    
\item  $n = 1111111000000000001$ , where the 1's and 0's are in highly ordered state.
\end{itemize}

\subsection{A Random Experiment}
Empowered by this newfound idea that,  $n+ (b - n \; mod \; b )$
is nothing but a random seed when added to  $bn$
to randomize the abstract machine's tape, we tried to experimentally verify the same.
A program was written where we replaced the generalized Collatz function by

$$
f(n) = bn+ R[n]
$$

where $R[n]$ is a random integer divisible by $b$ , and both $n$ and $R[n]$ 
(when represented in the base $b$) share the same number of digits. 
The assumption of having same number of digits was made to be consistent with the  original Collatz problem where $R[n]= n+1$.

We tested this hypothesis with base-2, and the results were positive. 
After a long but finite run, numbers do converge to 1, with generalized randomized Collatz sequence. 
We also note that adding a fixed even numbers in range $R[n] \in \{2,4,8,..., 2^{k-1}\}$ at every step, 
destroys the chaotic pattern, and makes it uniformly convergent.
This was tested with the largest even number in the set, $2(2^{k-1} -1)$. 

\subsection{A Binary Pseudorandom Digit Generator}
To show that the binary digits generated by the Collatz sequence are bitwise random, we wrote a Pseudorandom
digit generator based upon Collatz sequence.

\lstset{
         basicstyle=\footnotesize\ttfamily, 
         %numbers=left,               % Ort der Zeilennummern
         numberstyle=\tiny,          % Stil der Zeilennummern
         %stepnumber=2,               % Abstand zwischen den Zeilennummern
         numbersep=5pt,              % Abstand der Nummern zum Text
         tabsize=2,                  % Groesse von Tabs
         extendedchars=true,         %
         breaklines=true,            % Zeilen werden Umgebrochen
         keywordstyle=\color{red},
                frame=b,         
 %        keywordstyle=[1]\textbf,    % Stil der Keywords
 %        keywordstyle=[2]\textbf,    %
 %        keywordstyle=[3]\textbf,    %
 %        keywordstyle=[4]\textbf,   \sqrt{\sqrt{}} %
         stringstyle=\color{white}\ttfamily, % Farbe der String
         showspaces=false,           % Leerzeichen anzeigen ?
         showtabs=false,             % Tabs anzeigen ?
         xleftmargin=17pt,
         framexleftmargin=17pt,
         framexrightmargin=5pt,
         framexbottommargin=4pt,
         %backgroundcolor=\color{lightgray},
         showstringspaces=false      % Leerzeichen in Strings anzeigen ?        
 }
 \lstloadlanguages{
         Java
 }
   
\begin{lstlisting}

/******************************************************************
 * Just run Collatz Sequence and store the generation of numbers in a 
 * byte array, noting that we only store the result of 3/2 * ( n+1) -1 
 * Also - we do not go to the end of Convergence
 * Because then it would be too predictable.
 * What is apparent Collatz Sequence has 3 domains
 * 1. Randomization - where the input number would be shuffled and randomized
 * 2. Sustained Run - where the randomized tape would endure chaotic behaviour
 * 3. Convergence - where it would actually converge to 1.
 * In this generator - we want to operate only on [1] and [2] 
*******************************************************************/
static byte[] DoBinaryCollatzUntilSmall(BigInteger nB)
    {
        byte[] arr = new byte[0];
        double perc = nB.toByteArray().length;

        byte[] tmp =null;
        /* The 0.8 is the amount of reduction, bytewise. 
        * If we let it reduce to 1 byte, then the end bytes would be too predictable as all numbers 
        * converge to 1.
        * The real random parts are on the higher number of bytes range.
        * So, we choose a parameter and let the number of bytes reduced to 0.8 of the original.
        */
        int reduction = (int)Math.ceil(perc*0.8); 
        BigInteger radixB = new BigInteger("2");
        do
        {
            while(nB.mod(radixB).compareTo(BigInteger.ZERO) == 0)
            {
                  nB = nB.divide(radixB);
                  /*
                   * Why this is done?
                   * Assume we have generated X000, then the system would
                   * generate X00X0X, clearly pattern of X.
                   * with the new code we will not ever generate X again.
                   * if 3(X+1) is Y0, then we would generate
                   * XY
                   */
            }
            
            BigInteger mod = nB.mod(radixB);
            mod = radixB.subtract(mod);
            mod = mod.add(nB);
            nB = nB.multiply(radixB).add(mod);
            nB = nB.divide(radixB);
            
            tmp = nB.toByteArray();
            arr = add_to_array(tmp,arr);
           //System.out.println(tmp.length +":"+ reduction);

         }while( tmp.length > reduction );
         
         return arr;
     }
\end{lstlisting}

The findings are summarized  in the  table [\ref{table:pr-results}].

\begin{table}[!ht]
\caption{Results of the Diehard Tests }
\centering
\begin{tabular}{ | l | l |  l |  }
\hline\hline
Test Name                             & Resulting p-value[s]       &  End Result  \\ [0.5ex]
%heading
\hline
Birthday Spacing                      & 0.913665                                & PASSED     \\ \hline
Overlapping 5-Permutation             & 0.479927,0.631744                       & PASSED     \\ \hline
Binary Rank Test (31X31) Matrix       & 0.337549                                & PASSED     \\ \hline
Binary Rank Test (32X32) Matrix       & 0.568434                                & PASSED     \\ \hline
Binary Rank Test (6X8) Matrix         & 0.963464                                & PASSED     \\ \hline
The Bit Stream Test                   & min 0.05492, max 0.9487                 & PASSED     \\ \hline
OPSO                                  & min 0.0492,  max 0.9579                 & PASSED     \\ \hline   
OQSO                                  & min 0.0496,  max 0.9801                 & PASSED     \\ \hline  
DNA                                   & min 0.0631,  max 0.9877                 & PASSED     \\ \hline    
Count The 1's  On Byte Stream         & 0.980913, 0.589621                      & PASSED     \\ \hline
Count The 1's  On Specific Bytes      & min 0.015312,max 0.99139                & PASSED     \\ \hline
Parking Lot                           & 0.177727                                & PASSED     \\ \hline
Minimum Distance Test                 & 0.833870                                & PASSED     \\ \hline
3D Sphere Test                        & 0.954914                                & PASSED     \\ \hline
Squeeze Test                          & 0.483180                                & PASSED     \\ \hline  
Overlapping Sums Test                 & 0.491630                                & PASSED     \\ \hline
Runs Test Up                          & 0.256061, 0.528445                      & PASSED     \\ \hline
Runs Test Down                        & 0.485818, 0.751651                      & PASSED     \\ \hline
Craps Test                            & wins:0.682508,throws/game:0.268146      & PASSED     \\ \hline

\hline
\end{tabular}
\label{table:pr-results}
\end{table}

\subsection{General Collatz Sequence As A Random Walk Problem}

The General Collatz Sequence becomes a random walk problem when we see the results of the random bit generator.
We define the problem as follows:- \\
Define an one dimensional discrete space, where a point: $P(x) \in \{ 1,2,3,4,...\}$.
Assume `$k$' is the number of digits in the binary representation of a number, $n$.
Clearly $k \in \{ 1,2,3,4,...\}$ and can be represented as a point in that discrete space.
All the numbers $ \{ x : |x|=k \} $  
(where $|x|$ denotes number of digits of ``x'' when expressed in binary) 
maps to the same point ($k$) in that space.

If we define the origin of the space as ``1'', then the distance ``$D(n)$'' of the number with digits``$k$'' from origin is:  $D(n)=k-1, \; |n| = k $.

\begin{itemize}
\item If by any transform the number of digits increases for a number `n' , the distance $D$ increases.
\item If the number of digits remains same for the number `n' , the distance $D$ remains same.
\item If by any transform the number of digits decreases for a number `n' , the distance $D$ decreases.
\end{itemize}

Adding any random even number from interval $0$ to $2(2^{k -1} -1)$ would do either of the 3 things as follows :-
\begin{list}{Action \arabic{qcounter}:~}{\usecounter{qcounter}}
\item 	Effective Left-shift by adding $1$ to the left side, and the number of digits moves to $k+1$
\item	No movement i.e.  $k$ remains as it is.
\item	Right-Shift, as effective $0$s would be padded to the right, and the number of digits moves to less than $k$.
\end{list}

Point to be noted here is that for left-shift, we have step-size as only $1$. But for right-shift, the available step size is in the set $\{1,2,3,..,k-1\}$.
So, the random walk is so defined as $D(t+1) = D(t) + X$ where 
$$
X \in \{ 1,0,-1,-2,...,-(k-1)\}
$$  
``X'' being a random variable.
  
The logarithm of a number is nothing but the amount of digits needed to represent it.
In that case stopping time is nothing but the time taken to reach some lower number of digits, and can be taken as 
\begin{equation}\label{stopping_time}
\sigma(n)  \approx \frac{log_b(n)}{E(X)} 
\end{equation}
where $E(X)$ is statistical expectation of random variable ``$X$''. 

This is the reason why runtime of Collatz like sequences are known to have stopping time that is proportional to 
the logarithm of the number as mentioned in Lagarias \cite{lagarias} and comes from the 1st principle, rather than
any heuristic presented in Crandall \cite{crandall} and Shanks \cite{shanks}, and provides a theoretical
basis for further study.

We can also expect then, that if we do random sampling
on numbers with a very high number of digits 
(typically 500+ on smaller bases, and 100+ on larger bases) 
then take ratio of the average of the empirical stopping time
with the number of digits, the result would be approximately a fixed ratio, no matter how large we make the number of digits.The table [\ref{table:Collatz_Shift_Table}] shows the expectation values remains fixed, 
as number of digit goes large.
In other words, we can empirically find expectation $E(X)$ by this formula

$$
{E(X)} \approx  \frac{log_b(n)}{\sigma(n)}
$$ 
and this would have a limit when digits go large. 
This behavior is actually seen when we try to analyze general Collatz behavior. 
That experiment is the topic of the next section, as the limiting fixed ratio is what is actually observed.

We see that the convergence can be thought about as ``Classical Ruin'' problem, where you started with some fixed amount of digits, and start either winning 1, or loosing or staying same with probabilities $p_l$ and $p_r$ and $p_n$.

If the expectation is negative then it would ensure that after a random run, the system always would
end up with 1 or more digit less than how it started. Hence \emph{starting from $k$ digits, one would always
end in $k-i, i \in \{1,2,3,...\} $ eventually}.

\section{Generalized Collatz Sequence At Large Number of Random Digits}

We define the Probability $\displaystyle\lim_{x\to\infty}P(b,x,s)$ as the probability of moving to $s$ steps when 
number of digits ``x'' is approaching infinity, where $s \in \{1,0,-1,-2,...\}$.

For brevity we can write $\displaystyle\lim_{x\to\infty}P(b,x,s) = P(b,s)$ by dropping the ``x'' altogether.

Collatz sequence does not sample the Most Significant digits uniformly, because whenever it left-shifts, 
out of the possibility of 
$$
\{1,2,3,...,b-1\} \times \{0,1,2,3,...,b-1\}
$$ 

it can expand only to $(10)_b$ as the two most significant bits.  

We did experiments to find the probability of left-shift and right-shift of the Collatz system.
They were done using random digit stream to be feed to the Generalized Collatz sequence as integer input.

The convergence points are never more than two, even on 1000 or more digit long integers.
For example, base 3 has $\{1,7\}$ convergence points, and even at 1000 digit long integer, 
the system still converge at only these two points.

Upon fixing the number of digits, 50 sample numbers were created by choosing each digit randomly.
On these 50 numbers the left, neutral (remain on the same digit) and right shift probabilities were observed and averaged. 

The rationale of choosing each digit random is to ensure that the Collatz system does not have to do the randomization.
  
The table [\ref{table:Collatz_Shift_Table}] shows the system dynamics on different bases, and in different number of digits.
(They might add up to more than 1.00 as we have chosen for all the values - most occurring ones).

\begin{table}[!ht]
\caption{ P(b,s) Empirical Data - Showing Probability Left/Right  Movement  }
\centering
\begin{tabular}{ | l | l |  l | l | l | l | }
\hline\hline
Base       & Digits            & P(b,1)        & P(b,0)       & P(b,-Any)      & Expectation \\ [0.5ex]
%heading
\hline
2           & 1000 - 5000        & 0.39090        &  0.27657     & 0.33252    & -0.27149		 \\ \hline
3           & 500 - 1000         & 0.19704        &  0.55520     & 0.24775    & -0.17518		 \\ \hline
5           & 500 - 1000         & 0.09451        &  0.73837     & 0.16678    & -0.11436		 \\ [1ex]
\hline
\end{tabular}
\label{table:Collatz_Shift_Table}
\end{table}

One key thing to be noted here is that we modified the original sequence.
When a number of form $n = pb^s; p \; mod \; b \ne 0$ comes, instead of treating reduction of this number to 
``s'' steps into ``p'', we treated it as a single step reduction, an ``s'' step right shift.

Hence, there exist a P(b,-s), that is probability of occurrence of a number of the form $n = pb^s$.   

Another table [\ref{table:Collatz_Right_Prob_Table}] shows the P(b,-s) in different bases.
\begin{table}[!ht]
\caption{ P(b,-s) Empirical Data - Showing Probability Of Various Right Shift Movements  }
\centering
\begin{tabular}{ | l | l | l | l |  }
\hline\hline
Base            & P(b,-1)       & P(b,-2)     & P(b,-3)      \\ [0.5ex]
%heading
\hline
2              & 0.16699        &  0.08358     & 0.04117    \\ \hline
3              & 0.16500        &  0.05526     & 0.01851    \\ \hline
5              & 0.13310        &  0.02706     & 0.00527    \\ [1ex]
\hline
\end{tabular}
\label{table:Collatz_Right_Prob_Table}
\end{table}
\clearpage

\section {Collatz Probability Formulae For Base 2 - The Original Collatz Sequence}

In the earlier section, we have presented the probabilities as coming from the data itself. We show how they are to be 
expected from the 1st principle.

\begin{theorem}\label{right-shift-prob}
\textbf{Total Right Shift Probability.}

Total Right Shift Probability is : 
$$
\frac{1}{b+1} .
$$
\end{theorem}

\begin{proof}[Proof of Total Right Shift Probability Theorem]

The Collatz system is 3 state system. It can either do a left shift (increasing digits), remain neutral
(no change in the number of digits), or do a right shift (decrease in 1 or more number of digits).

This is a discrete step process.
Assume that at the $k$'th step the probability
of right shift is $p_r(b,k)$. So the probability that the system \emph{did not do a right shift at k'th step } is
given by:- 

$$
1 - p_r(b,k)
$$ 

Now, assuming Collatz Transform generates the lower significant digits in Uniform distribution, 
the probability of having one or more trailing zero would be digit $p( d_0 = 0 ) = \frac{1}{b}$,

which means, probability of having trailing zeros at $k+1$'th step is:-

$$
p_r(b,k+1) = (1 - p_r(b,k)) \times  \frac{1}{b}
$$

Now, at equilibrium, we would expect that

$$
p_r(b,k) \approx p_r(b,k+1) \approx p_r(b)
$$

Solving for $p_r(b)$ immediately gives:- 
 
\begin{equation}\label{right-shift-prob}
p_r(b)  =  \frac{1}{b+1}
\end{equation}

This establishes the theorem.
\end{proof}

\begin{theorem}\label{right-shift-prob-dist}

\textbf{Right Shift Probabilities.}
Right Shift Probabilities : P(b,-s) are geometrically distributed.
\end{theorem}

\begin{proof}[Proof of Right Shift Probabilities Theorem]
We show here that the P(b,-s) has the formula:-
\begin{equation}\label{right-shift-dist}
P(b,-s)  = \frac{(b-1)}{(b+1)}b^{-s} 
\end{equation}

We start with noting that P(b,-s) is actually the probability of having ``s'' numbers of ``0'' in 
the end of a number, feed into the right shift operation of the abstract machine.
Here we note that 
$$
P(b,-k-1) = \frac{1}{b}P(b,-k)
$$ 
as, to generate a ``k+1'' times ``0'' padded number from ``k'' times ``0'' padded number is to add an additional ``0''
to the right, which only 1 out of $\{0,1,2,3,..,b-2,b-1\}$ or $b$ possibilities.

We now deduce 
$$
P(b,-1) = \frac{b-1}{b+1}b^{-1}
$$ 
Assume a number $X=d_{n-1}d_{n-2}...d_1d_0$. It would undergo a single right shift iff:-
The system is doing a right shift, with $d_1$ is nonzero.
Formally:-
$$
P(b,-1) = P(d_1 \neq 0 \;  AND \;  `System \; Doing \; Right \;  Shift')
$$
as they are independent:-
$$
P(b,-1) = P(d_1 \neq 0 ) \times p_r(b)
$$ 
However, we clearly know that:-
$$
P(d_1 \neq 0 ) = \frac{b-1}{b}
$$ 

Hence, P(b,-1), that is only a single right shift takes place becomes:-
$$
P(b,-1) =  \frac{b-1}{b} \times \frac{1}{b+1}  =  \frac{b-1}{b+1}b^{-1}
$$

This establishes the theorem, when we note down that $P(b,-k-1) =\frac{1}{b}P(b,-k)$.
 
This is also clear from the table [\ref{table:Collatz_Right_Prob_Table}]
\end{proof}

We now note a population of original Collatz sequences, defined to be 
$$
\{...,f^i(x),...,f^n(x)\}
$$
We note that in any $f^i(x)$, the most significant 2 digits can be either ``10'' or ``11''.
But as the system always generates every left shift ``10'', hence we can be sure that in the population, ``10''
would be higher than that of ``11''. Assume also that we are concerned with only the digit expansion operation,
as digit contraction operation (right shift) does not change the MSBs.
We now derive the proportion of time the system would be in ``11'' and ``10'' states.

We note that the system prefixing ``11'' would always become ``10'' in the next iteration.
While,  ``10'' would  become ``11'' if an only if there is a carry from the lower digits.

We now deduce the probability of the carry.

\begin{theorem}\label{carry_theorem}
\textbf{Probability of carry to the MSB is a constant for Collatz Sequence.}

Probability of carry to the MSB is : 
$$
\frac{1}{3} .
$$  
\end{theorem}

\begin{proof}[Proof of Carry Theorem]

If the number can be written as $d_{k-1}d_{k-2}...d_0$ we can effectively write it as 
$$
d_{k-1}d_{k-2}...d_0 = d_{k-1}d_{k-2}X
$$
where X is itself a number which is NOT divisible by 2.
Hence, \emph { X is a max $k-2$ digit odd number, there is no restriction of initial 0's on X}. 

Now we note that for a combination of $(d_{k-1},d_{k-2}) = (1,0)$, the resultant number looks like
$$
10X0 + 10X + 1
$$

That is to say, to influence the $d_{k-1} = 1$, the $3X+1$ number had to have $|X| + 2$ digits. 

If we consider odd numbers X, starting from 1 to $(2^{k} -1)$, 
and check how many times applying generalized Collatz transform  increased the number of digits by 2,
we would reach the number of cases a carry is generated.

Assume the number is $d_{k-1}d_{k-2}X_{k-2}$.
If $(d_{k-1},d_{k-2}) = (1,1)$, there surely would be a carry.
If $(d_{k-1},d_{k-2}) = (1,0)$, there would be a carry \emph {iff} $X_{k-2}$ has one carry.
Hence comes the recurrence relation:-

\begin{equation}\label{rr_of_carry}
p_c(k) = \frac{1}{4} +  \frac{1}{4}p_c(k-2) 
\end{equation}

It is not hard to see that the above recurrence relation, at large k, would yield:-

$$
p_c(k) \approx p_c(k - 2) 
$$

And from there:-

\begin{equation}\label{probability_of_carry}
P(Carry) = p_c \approx \frac{1}{3} 
\end{equation}
And this completes our proof.

\end{proof}

\begin{theorem}\label{density_theorem}
\textbf{The probability of `10' and `11' states are constants.}

The probability of occurrence of `10' and `11' states are 0.60 and 0.40 respectively.
\end{theorem}

\begin{proof}[Proof of State Density Theorem]
Let's define the symbol $\Delta_{x}N_{y}$ as earlier system was in `x' and now in `y', 
Clearly then, we have the following behavior:-
$$
P(10|11) =\frac{\Delta_{11}N_{10} (t+1)}{N_{11}(t)} = 1
$$ 
And, as ``10X0'' when added with ``10X'' would always generate ``11'' system, \emph {iff} there is no carry:-
$$
P(11|10) = \frac{\Delta_{10}N_{11}(t+1)}{N_{10}(t)} = 1 - p_c
$$ 

In equilibrium state:-
\begin{equation}\label{msb_delta_population_equality}
\Delta_{11}N_{10}(t) = \Delta_{10}N_{11}(t)
\end{equation}
Immediately the below equations follow:-

\begin{equation}\label{msb_population_density_11}
\frac{N_{11}(t)}{N_{10}(t) + N_{11}(t) } = \frac{ 1 - p_c}{1 + (1 - p_c)} \approx \frac{2}{5}  
\end{equation}

\begin{equation}\label{msb_population_density_10}
\frac{N_{10}(t)}{N_{10}(t) + N_{11}(t) } = \frac{1}{1 + ( 1 - p_c)} \approx \frac{3}{5}  
\end{equation}

And this completes our proof.

\end{proof}

We end the discussion by stating that the empirical data shows the ratio of population 
of $N_{11}$ and $N_{10}$ as 0.40667 and 0.5930 approximately.

Now we derive the left shift probability of the base 2 Collatz system - that is 
probability that the number of digits would be increased by 1.

\begin{theorem}\label{left-prob-theorem}
\textbf{Left Shift Probability of Collatz Sequence.}

Left Shift Probability of Collatz Sequence is  0.4 approximately.
\end{theorem}

\begin{proof}[Proof of Left Shift Probability Theorem.]

The system left-shifts only when the system has no '0's in the right.
So, assume the probability of right shift is :-
$$
P(Right\;Shift) = p_r
$$
then, the probability of left-shift of the system is
$$
P(Left\;Shift) = p_l = (1-p_r)(P_L(10)P(10) + P_L(11)P(11))
$$ 
where $P(10) \approx 0.6$ and $P(11)\approx 0.4$ are probabilities 
that system would be in ``10'' and ``11'' MSB state.
But then, 
$P_L(11) = 1$, that is system always would increase the number of digits, if it is odd.

Right shift probability from \eqref{right-shift-prob}
   
\begin{equation}\label{probability_rightshift}
p_r(2) = p_r = \frac{1}{2+1} \approx 0.33
\end{equation}

That would mean that 
\begin{equation}\label{probability_leftshift}
p_l \approx (1-p_r)( 0.4 + 0.6p_c ) \approx 0.6(1-p_r) \approx 0.40
\end{equation} 
While empirically found value was 0.39.
This establishes the theorem.
\end{proof}

Now we find the expectation value for base 2.

\begin{theorem}\label{expectation-value-theorem}
\textbf{Expectation of Collatz Random Walk is constant.}

Expectation of Collatz Random Walk is $-0.270$ approximately. 
\end{theorem}

\begin{proof}[Proof of Expectation Value]
We note down it as
$$
E(X) = 1 \times p_l + 0 \times p_n +  {\sum_{s=-1} ^{-\infty}}  \frac{(b-1)}{(b+1)}sb^{s}
$$
For b=2, that rightmost term becomes:-
$$
{\sum_{s=-1} ^{-\infty}}  \frac{(2-1)}{(2+1)}s2^{s} = -\frac{2}{3}
$$

\begin{equation}\label{expectation_2}
E(X) \approx 0.40 - 0.667 \approx -0.267 
\end{equation} 

This proves the theorem.
\end{proof}

\begin{theorem}\label{expected-runtime-theorem}
\textbf{Expected Stopping Time of Collatz Sequence would be constant times the number of bits.}

Expected runtime (Stopping Time) of Collatz Sequence would be approximately 3.7 times the number of bits.
\end{theorem}

\begin{proof}[Proof of Expected Stopping Time]
The expected time to converge for random digits would be

\begin{equation}\label{converge_time_2}
 \sigma_2(n) \approx \frac{log_2(n)}{E(X)} \approx 3.74log_2(n)
\end{equation} 

\end{proof}

This is the linearity with respect to $log(n)$ we see in the large numbers (n).
Empirically observed value is $\sigma_2(n) \approx 3.60log_2(n)$.

\section{Summary And Further Work}

We can surely say that the Collatz like sequences are built on top of chaos generated 
by the implicit pseudo random generator,  $bn + n + (b - n \; mod \; b )$ according to (\ref{Collatz_lcg}).
We also noted that adding fixed even numbers in range (even at random)
$\{2,4,..., 2(2^{k -1} -1) \}$ at every step, destroys the chaotic pattern.

We have shown empirically that the pure random sequence terminates, 
and we proved that the statistical convergence exists by enumerating the possibility of the left and right shift and that
the resulting expectation is always negative. 
We also found numerous theorems from the basic probability principles, which shows
remarkable accuracy with experimental data.

Research is much needed on this specific type of pseudorandom sequence generator, and predictability of the bases, where
the system would converge to ``1'', and how many attractor/convergence points the system should have.

We end with a quote from the great Von Neumann:\emph{``Anyone who considers arithmetical methods of producing random digits is, of course, in a state of sin.''}. 
We humbly beg to differ, by stating that \emph{The Generalized  Collatz Sequence} can be used as a 
fantastic random symbol generator, passing all but one diehard tests (fails the minimum distance test) and almost all NIST tests. As Simplified Generalized Collatz Sequence (\ref{generalized_simplified_Collatz}) does not need any arithmetic operation other than addition, it is a good candidate for pseudorandom sequence generation in computers.

\section*{Acknowledgments}
The authors thank the colleagues of Microsoft India.

\end{document}